\documentclass[letterpaper, 10 pt, conference]{ieeeconf}  

\IEEEoverridecommandlockouts                              

\usepackage{tikz}
\usetikzlibrary{shapes,arrows}

\newtheorem{theorem}{Theorem}[section]

\newtheorem{lemma}[theorem]{Lemma}
\usepackage{amsmath} 
\usepackage{amssymb}

\usepackage{ltl} 

\usepackage{algorithm2e}

\usepackage{color}
\usepackage{subfig}
\usepackage{tikz}
\usetikzlibrary{arrows,fit,shapes,automata}
\usetikzlibrary{positioning,fit,calc,shapes}
\usetikzlibrary{decorations.fractals}
\usetikzlibrary{decorations.markings}

\newcommand{\defeq}{:=}

\newcommand{\dist}[1]{\mathcal{D}(#1)}
\newcommand{\ie}{\textit{i.e.}\xspace}

\newcommand{\part}{\mathcal{Q}}

\newcommand{\AP}{\mathcal{AP}}
\DeclareMathOperator*{\argmin}{arg\,min}

\def\fat#1{#1}

\title{Transfer Entropy in MDPs with  Temporal Logic Specifications
}

\author{Suda Bharadwaj \and Mohamadreza Ahmadi \and
 Takashi Tanaka \and Ufuk Topcu
\thanks{All authors are with the University of Texas at Austin. E-mail: \{suda.b, mrahmadi, ttanaka, utopcu\}@utexas.edu}%
}

\tikzstyle{block} = [draw, fill=blue!20, rectangle, 
    minimum height=3em, minimum width=6em]
\tikzstyle{sblock} = [draw, fill=blue!20, rectangle, 
    minimum height=3em, minimum width=3em]
\tikzstyle{sum} = [draw, fill=blue!20, circle, node distance=1.2cm]
\tikzstyle{input} = [coordinate]
\tikzstyle{output} = [coordinate]
\tikzstyle{pinstyle} = [pin edge={to-,thin,black}]

\begin{document}

\maketitle
\thispagestyle{empty}
\pagestyle{empty}


\begin{abstract}
Emerging applications in autonomy require  control techniques that take into account uncertain environments, communication and sensing constraints, while satisfying high-level mission specifications. Motivated by this need, we consider a class of Markov decision processes (MDPs), along with a \emph{transfer entropy} cost function. In this context, we study high-level mission specifications as co-safe linear temporal logic (LTL) formulae.  We provide a method to synthesize a policy that minimizes the weighted sum of the transfer entropy and the probability of failure to satisfy the specification. We derive a set of coupled non-linear equations that an optimal policy must satisfy. We then use a modified Arimoto-Blahut algorithm to solve the non-linear equations. Finally, we demonstrated the proposed method on a navigation and path planning scenario of a Mars rover.
\end{abstract}


\section{Introduction}
Autonomous systems are expected to deliver increasingly complex missions in dynamic and uncertain environments. In space applications, for example, these systems are in addition  fettered by communication or sensing restrictions. For instance, in the upcoming  Mars 2020 rover mission,  a Mars rover is tasked to safely explore an uncertain environment  and coordinate with a scouting helicopter~\cite{landau2015helicopter}. Missions of such sophisticated nature will necessitate on-board autonomy \cite{francis2017advanced,estlin2007increased}. Nonetheless, tight sensing constraints, due to the power consumption of on-board sensors and transmitters, and  bandwidth limitation on data sent from the Earth and orbiting satellites \cite{sherwood2014,Backes1999} further complicates the navigation task. In these cases, it is necessary  for autonomous agents to make decisions to complete their task with \emph{limited information}. 

Markov decision processes (MDPs) are one of the most widely studied models for decision-making under uncertainty in the fields of artificial intelligence, robotics, and optimal control \cite{Papadimitriou87}. We model the interaction between autonomous agent and an uncertain environment using a Markov decision process~(MDP) with an additional \emph{transfer entropy} cost that we refer to as a transfer entropy MDP~\cite{takashi17}. We use the additional transfer entropy cost~\cite{schreiber2000} to quantify the directional information flow  from the state of an MDP (representing the uncertain environment or the location of the autonomous agent) to the control policy. Intuitively, minimizing the transfer entropy promotes policies that rely less on the knowledge of the current state of the system. In communication theory, a related quantity called \emph{directed information} has been used to measure channel capacities in feedback systems \cite{massey1990causality,tatikonda2009capacity} as well as a proxy for feedback data rate to controllers \cite{silva2011achievable}. 


There has been significant work on quantifying information requirements for low-level control requirements, such as stability~\cite{Nair07}. However, quantifying information requirements for high-level decision-making scenarios that we are interested in are not as widely studied. There have been model reduction techniques for MDPs under temporal logic constraints studied where states and actions that are completely irrelevant to the mission are removed~\cite{Bharadwaj17,brazdil2014verification,ciesinski2008reduction}. However, these approaches do not \emph{quantify} the information flow to the controller from the state. \cite{Tishby2011} examines directed information in MDPs to quantify information and policies are penalized if they vary too much from a completely uninformed starting point, e.g, take any action with equal probability. We are, on the other hand, interested in studying the causality of information from the state to the controller, i.e, we seek to penalize sending information that is not relevant for the decision-making process. Hence, transfer entropy is a more suitable information-theoretic metric than directed information in our setting. 

%

%


We formally describe high-level mission specifications that are defined in temporal logic. Temporal logic has been used as a formal way to allow the user to relatively intuitively specify high-level specifications, in for example, robotics and autonomy applications~\cite{Svorenova2013,LPH15,wu2017learning}. 
Several tools exist to synthesize policies in MDPs with probabilistic temporal logic specifications~\cite{Fu15}. We study the effect of information restriction on satisfying temporal logic objectives in MDPs with a transfer entropy cost. 



\paragraph*{\textbf{Contributions}} We develop a novel framework to formally connect information-theoretic techniques for policy synthesis in MDPs with techniques from formal methods and probabilistic model checking. Specifically, our contributions are as follows:\\
(1) We develop a framework based on MDPs with a transfer entropy cost which places a cost on state variables that are `expensive to observe' in an information-theoretic sense. \\
(2) We incorporate a temporal logic constraint by optimizing the weighted sum of the probability of satisfying a mission specification and the transfer entropy cost.\\
(3) In contrast to standard MDP policy computation under temporal logic specifications, the transfer entropy cost leads to randomized optimal policies \cite{tanaka2017lqg,Todorov09,takashi17} necessitating policy search in an infinite state space. To solve this efficiently, we exploit a necessary optimality condition that the policy must satisfy.\\
(4) We solve these coupled non-linear equations using a modified version of an iterative algorithm from~\cite{Blahut72}.\\
(5) While the proposed method builds on earlier results in \cite{takashi17}, we generalize the setting to penalize subsets of state variables and incorporate temporal logic constraints. \\
(6) We apply our results in a case study involving path planning for a Mars rover.

\section{Preliminaries}
The sequence $(x_0,x_{1}...x_t)$ is denoted $x^t$ and the subsequence $x_l,x_{k+1}...x_k$ is denoted by $x_{l}^{k}$. We use upper-case letters to denote random variables and lower-case letters for the realizations of the corresponding random variable. \\
We denote by $\dist{\mathcal{X}}$ the set of all probability distributions on a finite
set $\mathcal{X}$, \ie all functions $f: \mathcal{X} \to [0,1]$ such that $\sum_{x\in \mathcal{X}}f(x)=1$. Finally, for a set $\mathcal{S}$, we define $2^\mathcal{S}$ as the set of all subsets of $\mathcal{S}$ and $\mathcal{S}^{\omega}$ as the set of all infinite sequences of elements in $\mathcal{S}$
\subsection{Markov Decision Processes}
\paragraph*{Labeled Markov decision process (MDP)} Consider a set $\AP$ of \emph{atomic propositions} which can be used, for example, to mark a state as being a ``faulty configuration'' (reaching it is, thus, undesirable), for example an obstacle. A \emph{labeled MDP} is an MDP whose states are labeled with atomic propositions. More formally, it is a tuple $M=(\mathcal{X},\mathcal{U},p,\AP,L)$ where 
\begin{itemize}
\item $\mathcal{X}$ is a
	finite set of \emph{states},
\item $\mathcal{U}$ is a finite alphabet of \emph{actions},
\item 	$p: \mathcal{X}\times\mathcal{U} \to \mathcal{D}({\mathcal{X}})$ is a \emph{probabilistic	transition function} that assigns, to a state $x\in \mathcal{X}$ and an action $u \in\mathcal{U}$, a probability distribution over the successor states. We abbreviate $p(x_{t},u)(x_{t+1})$ by $p(x_{t+1}|x_t,u_t)$.
\item $L : \mathcal{X} \rightarrow 2^{\AP}$ is the \emph{labeling function} which indicates the set of atomic propositions which are true in each state of the MDP.
\end{itemize}

\paragraph*{Runs and policies}
A \emph{run} from state $x_0$ with time horizon $T$ is a sequence $\rho = x_0 u_0 x_1 u_1 \dots ,x_{T-1},u_{T-1},x_{T}$ of states and actions such that for all $0 \leq t\leq T$ we have $p(x_{t+1}|x_t,u_t)>0$. 
A \emph{policy} corresponds to a way of selecting actions based on the history
of states and actions. While \emph{deterministic stationary} policies
are known to be sufficient for certain classes of problems, such as pure reachability ~\cite{puterman2014}, policies in general can be non-deterministic and history dependent. In this paper, we consider the general form and formally represent a policy as a conditional probability distribution $q_t(u_t|x^t,u^{t-1})$. 

A run $\rho$ is \emph{consistent} with a policy $q$ if it can be
obtained by extending its prefixes using $q$. Formally, $\rho=x_0
u_0 x_1 u_1 \dots$ is consistent with $q$ if for all $t \ge 0$ we have that
$u_t \in \{u| q_t(u|x^t,u^{t-1} > 0)\}$ and $p(x_{t+1}|x_t,u_t)>0$

\paragraph*{Markov chain}
A Markov chain is a tuple $(\mathcal{X},x_I,p)$ where $\mathcal{X}$ is (in our case) a finite set of states, $x_I \in \mathcal{X}$ is the initial state, and $p: \mathcal{X} \to \dist{\mathcal{X}}$ is a probabilistic transition function. An MDP $M$ together with a policy $q$ induces a \emph{Markov chain} $M^q$.  Notions of runs in a Markov chain are the same as those defined earlier. 

Given a Markov chain $M^q = (\mathcal{X},x_I,p)$, the state visited at the step $t$ is
a random variable. We denote by $h^{k}(x,\mathcal{B})$ the probability that a
run starting from state $x$ visits the set $\mathcal{B}$ in exactly $k$ steps. By definition
$h^{\leq i}(x,\mathcal{B}) = \sum_{k=0}^{i} h^{k}(x,\mathcal{B})$ denotes the probability that run from $x$ reaches the set $\mathcal{B}$ in \emph{at most} $i$ steps where $h^0(x,\mathcal{B})$ is $0$ if $x
\not\in \mathcal{B}$ and $1$ otherwise.



\subsection{Temporal Logic}
\paragraph*{Co-safe linear temporal logic} We utilize linear temporal logic (LTL) to specify the objectives of the system. For example, we can specify that an agent infinitely often patrols a certain set of states (liveness) while not entering undesirable states (safety). For the formal semantics of LTL, see \cite{BaierKatoen08}. We are interested in minimizing the expected information cost over a finite time horizon. However, this is not well defined for general LTL formulas as the cost can, in general, diverge. We will thus look at a class of formulas that can be satisfied in finite time called co-safe formulas which we denote by $\varphi$. These are commonly used in optimal control of MDPs \cite{Lacerda14}. It was shown in \cite{kupferman2001model} that any LTL formula in which the negation is only applied directly to the atomic propositions called \emph{positive normal form} and which only uses the connectives $\LTLdiamond$ (eventually), $\LTLcircle$ (next), and $\LTLu$ (until) are co-safe. 

\paragraph*{Deterministic finite automaton (DFA)} Any co-safe LTL formula $\varphi$ can be translated to a DFA \cite{kupferman2001model}. A DFA is a tuple $\mathcal{A}_{\varphi} = (\mathcal{S},s_I,2^{\AP}, \delta,\textrm{Acc})$ where $\mathcal{S}$ is a finite set of states, $\AP$ is a set of atomic propositions, $2^{\AP}$ is the alphabet of the automaton. $\delta: \mathcal{S} \times 2^{\AP} \rightarrow \mathcal{S} $ is the transition function and $s_I \in \mathcal{S}$ is the initial state. The acceptance condition $\textrm{Acc}$ is an accepting set of states $\textrm{Acc} \subseteq S$. Since $\varphi$ is co-safe, it is known that all infinite sequences that satisfy $\varphi$ have a finite \emph{good prefix}. Let $w = w_0 w_1 \dots \in {({\fat{2}^{\AP}})}^{\omega}$ be an infinite word in the language of the automaton such that $w \vDash \varphi$, then there exists $n\in \mathbb{N}$ such that $w_0,w_1,\dots w_n \vDash \varphi$. Hence, after reaching an accepting state $s \in \textrm{Acc}$, we can 'complete' the prefix by setting $\delta(s,\alpha) = s$ for all $\alpha \in 2^{\AP}$


\paragraph*{Product MDP}
Given an MDP $M=(\mathcal{X},\mathcal{U},p,\AP,L)$ and a specification DFA
$\mathcal{A}_{\varphi} = (\mathcal{S},s_I,2^{\AP}, \delta,\textrm{Acc})$, we can define a \emph{product
MDP}, $\mathcal{M} \defeq M \times\mathcal{A}_{\varphi}$, as $\mathcal{M}
:= (\mathcal{V},\mathcal{U}, \Delta,v_0, L_{\varphi},\textrm{Acc}_{\mathcal{M}})$ where
\begin{itemize}
	\item $\mathcal{V} = \mathcal{X} \times \mathcal{S}$;
	\item $\Delta: \mathcal{V} \times \mathcal{U} \rightarrow \dist{\mathcal{V}}$ is a probabilistic function such that $\Delta\left((x_{t+1},s_{t+1})\vert (x_t,s_t)\right) = p(x_{t+1} \vert x_t,u_t ) $ if $\delta(s_t,L(x_{t+1}))
		= s_{t+1}$;
	\item $v_0 = (x_0,s_I)$; is the initial state;
	\item $L_{\varphi} = L(x) \cup \{\textrm{acc}_\varphi\}$ if $L(x) \in \textrm{Acc}$ and $L(x)$ otherwise; and
	\item $\textrm{Acc}_{\mathcal{M}}$ is the set of all states where the new atomic proposition $\textrm{acc}_\varphi$ is true.
\end{itemize}
Simply, once a run $\rho$ in $\mathcal{M}_{\varphi}$ reaches a state labeled with the atomic proposition $\textrm{acc}_\varphi$, it satisfies the formula $\varphi$. We denote a run $\rho$ as satisfying $\varphi$ by $\rho \vDash \varphi$. Hence, the problem of finding a policy $q$ that maximizes the probability of satisfying a given co-safe LTL specification becomes a matter of synthesizing a policy to reach a state in $\textrm{Acc}_{\mathcal{M}}$. This is a reachability problem in an MDP and can be solved using value iteration. This results in a \emph{memoryless} policy in $\mathcal{M}_\varphi$. Intuitively, the DFA component states of the product MDP can be thought of a \emph{memory state}. From this policy we can construct a \emph{finite-memory} policy in $M$. For more details on this construction, we refer the reader to  \cite{forejt2011automated}.

\section{Problem Statement}
In this section, we present the class of MDPs we consider and we formulate the problem under study.

Let $M=(\mathcal{X},\mathcal{U},p,\AP,L)$ be a finite labeled MDP.
Let $\mu_{t}(x^{t}, u^{t-1})$ be the joint distribution defined recursively by the state transition probability $p(x_{t+1}|x_t, u_t)$ and a policy $q_t(u_t|x^t, u^{t-1})$ as
\begin{align}
&\mu_{t+1}(x^{t+1}, u^t) \nonumber \\
&=p_t(x_{t+1}|x_t, u_t)q_t(u_t|x^t, u^{t-1})\mu_t(x^t, u^{t-1}). \label{eqmu}
\end{align}
A transfer entropy MDP is a labeled MDP with a split state space $\mathcal{X} = \bar{\mathcal{X}} \times \tilde{\mathcal{X}} $. Formally, transfer entropy MDP is a tuple $M=(\bar{\mathcal{X}},\tilde{\mathcal{X}},\mathcal{U},p,\AP,L)$, where  $\bar{\mathcal{X}}$ denotes the \emph{expensive} state variables, whereas  $\tilde{\mathcal{X}}$ denotes the \emph{free} state variables. We assume that the cost of information transfer from $\{\bar{X}_t\}$ to $\{U_t\}$ given $\{\tilde{X}_t\}$ over the time horizon $0\leq t\leq T-1$ is proportional to the (causally conditioned) \emph{transfer entropy} defined in \eqref{eqdi}.

\begin{equation}
\label{eqdi}
I(\bar{X}^{T-1}\rightarrow U^{T-1}\| \tilde{X}^{T-1})=\sum_{t=0}^{T-1} I(\bar{X}^t; U_t|U^{t-1},\tilde{X}^t).
\end{equation}

where $I (\bar{X}^t;U_t  \vert U^{t-1},\tilde{X}^{t})$ \cite{cover2012elements} is the \emph{conditional mutual information} and can be explicitly written as
\[
I  (\bar{X}^t;U_t \vert U^{t-1}\!\!,\tilde{X}^{t}) = \!\! \sum_{\mathcal{X}^t, \mathcal{U}^t} \mu_t(x^t\!, u^{t-1})\log\frac{q_t(u_t|x^t, u^{t-1})}{\nu_t(u_t|\tilde{x}^t, u^{t-1})}.
\]

and $\nu_t(u_t|\tilde{x}^t, u^{t-1})$ is the conditional distribution obtained by conditioning and marginalizing the joint distribution $\mu_{t}(x^{t}, u^{t-1})$. More specifically, 
\begin{equation}
\nu_t(u_t|\tilde{x}^t, u^{t-1})=\sum_{\bar{\mathcal{X}}^t}\mu_t(\bar{x}^t|\tilde{x}^t, u^{t-1})q_t(u_t|x^t,u^{t-1}).
\label{eqdefnu}
\end{equation}

We note that the notion of \emph{directed information} is introduced by \cite{massey1990causality} based on \cite{marko1973bidirectional}, and its generalization with causal conditioning by \cite{kramer1998causal}. Intuitively, \eqref{eqdi} can be understood as the information flow from a random process $\{\bar{X}_t\}$ to $\{U_t\}$ given $\{\tilde{X}_t\}$ as side information.

  
  To motivate this formulation, we present an example in which such a construction is natural.

Consider a Mars rover for the upcoming Mars 2020 mission~\cite{landau2015helicopter}.  Mars rovers have to complete their tasks in mostly unknown environments. Limited a priori knowledge of the terrain and possible obstacles can be provided from low-resolution satellite imagery. This information, however, is often not enough for decision-making as was evidenced by the Curiosity rover which suffered punctures, due to the unexpected presence of jagged, immobile, rocks embedded in the terrain. For the Mars 2020 mission, a helicopter has been proposed to act as a scout \cite{landau2015helicopter} to assist with planning. Figure \ref{fig:mars2020} shows an artists' rendering of the helicopter flying ahead to scout. The helicopter can then transmit information of the terrain back to the rover which is used for planning to satisfy the mission specification. 

\begin{figure}
\centering
\includegraphics[scale=0.36]{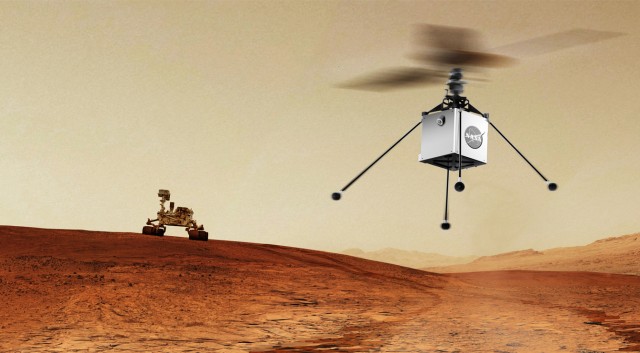}
\caption{Artist's rendering of the proposed helicopter to scout for the Mars rover. The helicopter can fly ahead and send information back to the rover about the presence of any obstacles. \cite{landau2015helicopter}.}\label{fig:mars2020}
\end{figure}

We model the dynamics of the rover in the Martian environment as an MDP with split state space $\mathcal{X} = \bar{\mathcal{X}} \times \tilde{\mathcal{X}}$. At every time step $t$, the component $\tilde{X}_t$ of the state vector is immediately available to the autonomous agent, e.g, from onboard sensors of the rover, while the component $\bar{X}_t$ is only available from a remote sensor, e.g, the scouting helicopter. We are thus interested in finding a policy $q_t(u_t|x^t,u^{t-1})$ that minimizes the information transfer from $\bar{X}$ to $U$. This information flow is captured by the transfer entropy cost. We can represent this system using a feedback control architecture shown in Figure~\ref{fig:NCS}. 

Additionally, the rover has to satisfy specification $\varphi$, given by a co-safe LTL formula is to a given threshold $0 \leq D \leq 1$ in the probability. Let
$\mathbb{P}_{q_t}^T(x_0 \vDash \varphi)$ be the probability of satisfaction of $\varphi$ by policy $q_t$ in finite time horizon T from initial state $x_0$. We define $J(X^{T},U^{T-1}) \defeq 1 - \mathbb{P}_{q_t}^T(x_0 \vDash \varphi)$ to be the probability of failure.


The main problem we study in this paper can be described as
\begin{align}
\min_{\{q_t(u_t|x^t,u^{t-1})\}_{t=0}^{T-1}} & I(\bar{X}^{T-1}\rightarrow U^{T-1}\| \tilde{X}^{T-1}) \nonumber \\
& \textrm{s.t } J(X^{T},U^{T-1}) \leq 1 - D. \label{eqn:constopt}
\end{align}

\begin{figure}[h]
\centering{
{\includegraphics[scale=0.34]{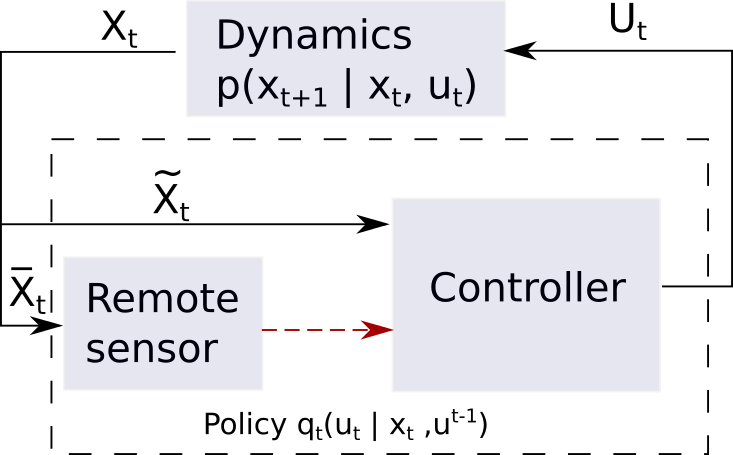}}
\caption{Example of a feedback control architecture with part of the statespace $\bar{X}_t$ being measured remotely. The red arrow indicates a band limited communications channel so transmissions are restricted.}
\label{fig:NCS}
}
\end{figure}

\section{Incorporating Temporal Logic Constraints}
In this section, we demonstrate how to take into account high-level mission specifications in terms of a co-safe LTL formula and cast the constrained control problem  into the form of optimization problem~\eqref{eqn:constopt}.


Consider a finite labeled MDP with transfer entropy cost $M=(\mathcal{\hat{X}},\mathcal{U},p,\AP,L)$ where, as before, the state space of $M$ is split into expensive and cheap to measure state variables $\mathcal{\hat{X}} = \mathcal{\bar{X}}_e \times \mathcal{\tilde{X}}_f$. We are additionally given a specification DFA $\mathcal{A}_{\varphi} = (\mathcal{S},s_I,2^{\AP}, \delta,\textrm{Acc})$, and finite time horizon $T$. The product transfer entropy MDP is $\mathcal{M}\defeq (\mathcal{V},\mathcal{U}, \Delta,v_0,L_{\varphi},\textrm{Acc}_{\mathcal{M}})$.  Hence, we will have the state space $\mathcal{V} = (\mathcal{\bar{X}}_e \times \mathcal{\tilde{X}}_f) \times S$. Now, for notational simplicity, we set $\mathcal{X} = \mathcal{V}$, the free to measure state $\mathcal{\tilde{X}} = (\mathcal{\tilde{X}}_f,\mathcal{S})$ (we assume without loss of generality that the state in the automaton is freely known), and the expensive to measure state $\mathcal{\bar{X}} = \mathcal{\bar{X}}_e$. Let $X = (\bar{X}_e,\tilde{X}_f,S)$ and $x = (\bar{x}_f,\tilde{x}_s,s)$ be defined similarly. Thus, our state space is now $\mathcal{X}= \mathcal{\bar{X}} \times \mathcal{\tilde{X}}$ with random variable $X = (\bar{X},\tilde{X})$. 

We define a state-action cost in the product MDP in the following way. We define a function $c_t(x_t,u_t,x_{t+1})$, such that for every transition from $x_t$ to $x_{t+1}$, the cost is $0$ if neither $x_t$ or $x_{t+1}$ are in $\textrm{Acc}_{\mathcal{M}}$. The cost is $-1$ if $x_t \notin \textrm{Acc}_{\mathcal{M}}$ and $x_{t+1} \in \textrm{Acc}_{\mathcal{M}}$ and no state in $\textrm{Acc}_{\mathcal{M}}$ has been visited prior to reaching $x_t$.  Intuitively, minimizing this quantity will result in a policy $q$ that maximizes the probability of reaching $\textrm{Acc}_{\mathcal{M}}$ and hence, equivalently will maximize the probability of satisfying the temporal logic specification in $M$. The expected accumulated reward from state $x_0$ given by $\sum_{t=0}^{T-1}\mathbb{E}\{c_t(x_t,u_t,x_{t+1})\}$ will equal the \emph{negative} of the reachability probability to the target set $C$ in $T-$steps \ie we have
\vspace{-0.2cm}
\begin{align}\label{eqn:cost}
\sum_{t=0}^{T-1}\mathbb{E}\{c_t(x_t,u_t,x_{t+1})\} = -h^{\leq T}(x,\textrm{Acc}_{\mathcal{M}}).
\end{align}
\vspace{-0.3cm}

Setting $J(X^T,U^{T-1}) = \sum_{t=0}^{T-1}\mathbb{E}\{c_t(x_t,u_t,x_{t+1})\}$, we obtain an equivalent formulation of \eqref{eqn:constopt} with cost function $c_t$ as defined earlier.

\paragraph*{Remark} The constrained optimization problem in equation (\ref{eqn:constopt}) can be written as a \emph{Lagrangian relaxation} in the following way

\begin{align}\label{eqmainproblem}
\min_{\{q_t\}_{t=1}^T}  J(X^{T},U^{T-1}) + \beta I(\bar{X}^T \rightarrow U^{T-1}||\bar{X}^T)
\end{align}
where $\beta$ is a positive constant.

Intuitively, this means that we want to minimize the information flow from the state variables in $\mathcal{\bar{X}}$ subject to the constraint on the accumulated cost $J$. Using the cost function defined in \eqref{eqn:cost}, this constrains the probability of not satisfying the specification. The rest of the paper will deal with \eqref{eqmainproblem} 




\section{Optimality Conditions}
In this section, we derive a necessary optimality condition for \eqref{eqmainproblem}.
The result in this section generalizes \cite{charalambous2014optimization, stavrou2015information} to conditional directed information. In the following derivation, we assume $\beta=1$ for simplicity.
First, we rewrite the objective function in \eqref{eqmainproblem} explicitly as a function of $q$ and $\nu$. Using the definition of the causally conditioned directed information \eqref{eqdi}, the objective function can be written in a stage-additive form as
$f(q_0, ... , q_{T-1}; \nu_0, ... , \nu_{T-1})=\sum_{t=0}^{T-1} \ell_t$ with
\begin{align*}
\ell_t=&\sum\nolimits_{\mathcal{X}^t}\sum\nolimits_{\mathcal{U}^t} \mu_t(x^t, u^{t-1})q_t(u_t|x^t, u^{t-1}) \\
&\times \left(  \sum\nolimits_{\mathcal{X}_{t+1}}p(x_{t+1}|x_t, u_t)c_t(x_t, u_t, x_{t+1}) \right. \\
& \hspace{12ex}+\left.  \log \frac{q_t(u_t|x^t, u^{t-1})}{\nu_t(u_t|\tilde{x}^t, u^{t-1})}\right)
\end{align*}
where $\mu_t$ is recursively defined by $q_t$ via  \eqref{eqmu}. 
To analyze our cost function $f(q; \nu)$ to be minimized, we note the following simple lemma, which is a straightforward generalization of \cite[Theorem 4(b)]{blahut1972computation}. 
\begin{lemma}
For fixed $q$, $f(q;\nu)$ is minimized by \eqref{eqdefnu}.
\end{lemma}
This lemma implies that, although $q$ and $\nu$ must satisfy \eqref{eqdefnu} (we write $\nu(q)$ to emphasize that $\nu_t$ for $0\leq t \leq T-1$ is a function of $q_t$ for $0\leq t \leq T-1$), the constraint \eqref{eqdefnu} will be automatically satisfied by solving $\min_{q, \nu} f(q; \nu)$.
In particular, if $q^*$ is an optimal solution to \eqref{eqmainproblem}, and if $\nu^*=\nu(q^*)$, then $(q^*, \nu^*)$ is an optimal solution to $\min_{q, \nu} f(q; \nu)$.
Since optimality of $(q^*, \nu^*)$ implies coordinate-wise optimality of $q^*$, this implies 
\begin{equation}
\label{eqqstarargmin}
q^* \in \argmin_q f(q;\nu^*).
\end{equation}
Thus, if $q^*$ is an optimal solution to \eqref{eqmainproblem}, it necessarily satisfies $\nu^*=\nu(q^*)$ and \eqref{eqqstarargmin} simultaneously. The next lemma shows that the optimal solution to the right hand side of \eqref{eqqstarargmin} can be obtained analytically.
\begin{lemma}
\label{lem52}
For fixed $\nu^*$, define sequences $\rho_t^*$ and $\phi_t^*$ for $0\leq t \leq T-1$ backward in time by
\begin{align*}
\phi_t^*(x^t, u^{t-1})=&\sum\nolimits_{\mathcal{U}_t} \nu_t^*(u_t|\tilde{x}^t, u^{t-1})\exp\{-\rho_t^*(x^t, u^t)\} \\
\rho_t^*(x^t, u^t)=&\sum\nolimits_{\mathcal{X}_{t+1}}p(x_{t+1}|x_t, u_t) \\
&\times \{c_t(x_t, u_t, x_{t+1})-\log \phi_{t+1}(x^{t+1}, u^t)\}
\end{align*}
with terminal condition $\phi_T^*(x^T, u^{T-1})=1$. Then, the optimal solution to $\min_q f(q; \nu^*)$ satisfies
\begin{equation}
\label{eqoptqt}
q_t^*(u_t|x^t, u^{t-1})=\frac{\nu_t^*(u_t|\tilde{x}^t, u^{t-1})\exp\{-\rho_t^*(x^t, u^t)\}}{\phi_t^*(x^t, u^{t-1})}
\end{equation}
$\mu_t$-almost everywhere for each $0\leq t \leq T-1$.
\end{lemma}
\begin{proof}
See Appendix \ref{sec:prf}
\end{proof}
The main result of this section is thus summarized as follows.
\begin{theorem}
An optimal solution $q^*$ to \eqref{eqmainproblem} necessarily satisfies the following set of nonlinear equations
\begin{subequations}
\label{eqoptcondition}
\begin{align}
\mu_{t+1}^*(x^{t+1}, u^t)=& \;p(x_{t+1}|x_t, u_t)q_t^*(u_t|x^t, u^{t-1}) \nonumber \\
&\hspace{7ex}\times \mu_t^*(x^t,u^{t-1})  \\
\nu_t^*(u_t|\tilde{x}^t, u^{t-1})=&\sum_{\bar{\mathcal{X}}^t}\mu^*_t(\bar{x}^t|\tilde{x}^t, u^{t-1})q_t^*(u_t|x^t,u^{t-1}) \\
\rho_t^*(x^t, u^t)=&\sum_{\mathcal{X}_{t+1}}p(x_{t+1}|x_t, u_t) \{c_t(x_t, u_t, x_{t+1}) \nonumber \\
&\hspace{7ex}-\log \phi_{t+1}^*(x^{t+1}, u^t)\} \\
\phi_t^*(x^t, u^{t-1})=&\sum_{\mathcal{U}_t} \nu_t^*(u_t|\tilde{x}^t, u^{t-1})\nonumber  \\
&\hspace{7ex}\times \exp\{-\rho_t^*(x_t, u^t)\} \\
q_t^*(u_t|x^t, u^{t-1})=&\frac{\nu_t^*(u_t|\tilde{x}^t, u^{t-1})\exp\{-\rho_t^*(x^t, u^t)\}}{\phi_t^*(x^t, u^{t-1})}
\end{align}
\end{subequations}
for each $0\leq t\leq T-1$ with the given initial condition $\mu_0^*$ and the terminal condition $\phi_T^*(x^T, u^{T-1})=1$.
\end{theorem}

\subsection{Forward-backward algorithm}
The optimality condition \eqref{eqoptcondition} is a set of coupled non-linear equations with respect to the variables $\mu^{*},\nu^*,\rho^*,\phi^*,q^*$. In order to solve these we propose a numeric forward-backward algorithm. Firstly, note that if $\rho^*,\phi^*,q^*$ are known, $\mu^{*},\nu^*$ can be solved forwards in time. Similarly, if $\mu^{*},\nu^*$ are known then the others can be solved backwards in time. 

To solve this, we do the following. First we make a guess for each of the variables. We then solve the forward-time equations for $\mu^{*},\nu^*$. We use these values to then solve for $\rho^*,\phi^*,q^*$ backwards in time. This process is repeated until convergence.  This can be viewed as a generalization of the Arimoto-Blahut algorithm \cite{Blahut72}.

\paragraph*{Remark} We note that the problem formulation and derived equations are infinite-history, \ie they depend on the state and control actions from $t=0$ to $t=T-1$. In order to make this computationally tractable to solve, we modify the algorithm to search for the best policy of the form $q_t^*(u_t|x_t,u^{t-1}_{t-n})$ with some finite $n$. 
We refer the reader to \cite{takashi17} for more details on the similar algorithm and its convergence results. 



\section{Numerical Results}\label{sec:exp}
We consider a scenario where the rover is tasked with collecting samples from a specific region. The environment is modeled as an MDP as motion can be stochastic, i.e, slippage can occur. The mission is specified as a co-safe LTL specification. 

We analyze two different case studies. In the first experiment the rover has to plan around a moving obstacle, but the knowledge of the location of the moving obstacle is penalized. In the second experiment, the rover has some a priori knowledge of the terrain, but there is a cost to using any additional information. 

\subsection{Moving obstacle}
We solve the motion planning problem under sensing constraints in a gridworld as shown in \ref{fig:casestudy}\subref{fig:exp1}. Consider a scenario where rover is tasked with reaching the goal state in green whilst avoiding collisions with the red static obstacles and an orange moving obstacle that moves in the area shown. For example, the helicopter can be completing a separate mission and we do not want the rover and helicopter to collide, but we also want to limit their communication to conserve power. Hence, we treat the helicopter as a moving obstacle and add an information cost to its position.

We express this in LTL as $\lnot$'crash' $\LTLuntil$ 'goal'. The atomic proposition 'crash' is true in the red static obstacles and when the state of the rover is the same as the state of the moving obstacle. The atomic proposition 'goal' is true in the green cell. The DFA representation is shown in Figure~\ref{fig:casestudy}\subref{fig:dra}.
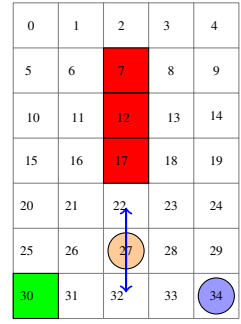
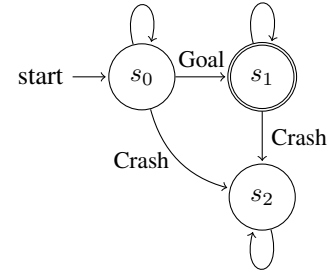
\begin{figure}
\subfloat[Gridworld with moving obstacle \label{fig:exp1}]{
\begin{tikzpicture}[scale=1.2]
\draw[step=0.5cm,color=gray] (-1.5,-2.5) grid (1,1);
\filldraw[fill=red,draw=black] (0,0.5) rectangle (-0.5,0);
\filldraw[fill=red,draw=black] (0,0) rectangle (-0.5,-0.5);
\filldraw[fill=red,draw=black] (0,-0.5) rectangle (-0.5,-1);
\filldraw[fill=green,draw=black] (-1.5,-2.5) rectangle (-1,-2.0);
\filldraw[fill=blue!40!white,draw=black] (+0.75,-2.25) circle (0.2cm);
\filldraw[fill=orange!40!white,draw=black] (-0.25,-1.75) circle (0.2cm);
\draw[blue,thick,->] (-0.25,-1.73) -> (-0.25,-1.27);
\draw[blue,thick,->] (-0.25,-1.73) -> (-0.25,-2.21);
\node at (-1.30,+0.75) {\tiny{0}};
\node at (-0.80,+0.75) {\tiny{1}};
\node at (-0.30,+0.75) {\tiny{2}};
\node at (0.20,+0.75) {\tiny{3}};
\node at (0.73,+0.75) {\tiny{4}};
\node at (-1.33,+0.25) {\tiny{5}};
\node at (-0.85,+0.25) {\tiny{6}};
\node at (-0.3,+0.25) {\tiny{7}};
\node at (0.25,+0.25) {\tiny{8}};
\node at (0.75,+0.25) {\tiny{9}};
\node at (-1.28,-0.27) {\tiny{10}};
\node at (-0.78,-0.27) {\tiny{11}};
\node at (-0.28,-0.27) {\tiny{12}};
\node at (0.28,-0.27) {\tiny{13}};
\node at (0.75,-0.25) {\tiny{14}};
\node at (-1.3,-0.75) {\tiny{15}};
\node at (-0.8,-0.75) {\tiny{16}};
\node at (-0.3,-0.75) {\tiny{17}};
\node at (0.25,-0.75) {\tiny{18}};
\node at (0.75,-0.75) {\tiny{19}};
\node at (-1.35,-1.25) {\tiny{20}};
\node at (-0.85,-1.25) {\tiny{21}};
\node at (-0.32,-1.25) {\tiny{22}};
\node at (0.25,-1.25) {\tiny{23}};
\node at (0.75,-1.25) {\tiny{24}};
\node at (-1.35,-1.75) {\tiny{25}};
\node at (-0.85,-1.75) {\tiny{26}};
\node at (-0.25,-1.75) {\tiny{27}};
\node at (0.25,-1.75) {\tiny{28}};
\node at (0.75,-1.75) {\tiny{29}};
\node at (-1.35,-2.25) {\tiny{30}};
\node at (-0.85,-2.25) {\tiny{31}};
\node at (-0.35,-2.25) {\tiny{32}};
\node at (0.25,-2.25) {\tiny{33}};
\node at (0.75,-2.25) {\tiny{34}};
\end{tikzpicture}
}
\subfloat[Specification DFA\label{fig:dra}]{
	\begin{tikzpicture}[shorten >=1pt,node distance=1.6cm,on grid]
	\node[state,initial]   (q_0)                {$s_0$};
	\node[state,accepting]           (q_1) [right=of q_0] {$s_1$};
	\node[state] (q_2) [below=of q_1] {$s_2$};
	\path[->] (q_0) edge                node [above] {\small Goal} (q_1)
	edge [loop above]   node         {} ()
	edge [bend right]   node [left] {\small Crash} (q_2)
	(q_1) edge                node [right] {\small Crash} (q_2)
	edge [loop above]         node {} (q_1)
	(q_2) edge [loop below]   node {} (q_2);
	\end{tikzpicture}
}
\caption{Gridworld and DFA  with $\textrm{Acc}_{\mathcal{M}} = (s_1)$ depicting a scenario where agent in blue has to reach green target cell without crashing into the red static obstacles or the orange moving obstacle.}\label{fig:casestudy}
\end{figure}

The rover has the choice of moving in 4 directions - North, South, East, and West or staying still. The motion is stochastic, i.e., it has a probability of slip. For example, if it chooses to move north, it has a probability to 'slip' (due to terrain effects like running sand) and move to a state north east or north west. 

The state space of the MDP is $(x,y,x_{obs},y_{obs})$ where $(x,y)$ is the position of the rover and $(x_{obs},y_{obs})$ is the position of the moving obstacle. We assume that the state of the moving obstacle to be expensive to observe. Formally, we let $\mathcal{\tilde{X}} = (x,y)$ and $\mathcal{\bar{X}} = (x_{obs},y_{obs})$. 

Since there is a probability to slip, the agent has a non-zero probability of crashing and not satisfying the specification if it goes the long way around the wall. If the agent knows the position of the moving obstacle at all times, it can plan to avoid collision, and hence the shorter path will have the higher probability of satisfaction. Intuitively, we expect to see if that we set the $\beta$ parameter high, \ie if the cost of information is high, the agent will go the long way around the wall as it will be too expensive to observe the moving obstacle. We use a time horizon $T = 25$ and test for $\beta = 0.5$ and $\beta=5$.

\begin{figure}
	\subfloat[$\beta = 0.5$\label{simple-grid}]{
	\includegraphics[scale=.5]{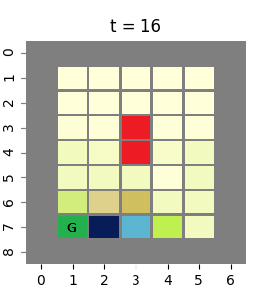}}
	\subfloat[$\beta = 5$\label{fig:simple-transitions}]{
		\includegraphics[scale=0.5]{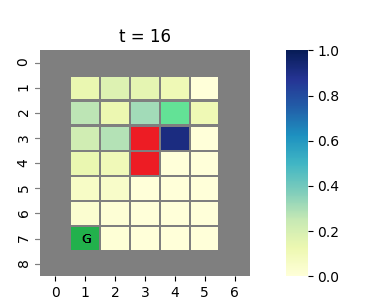}
	}
	\caption{Probability distribution of the agent after 16/25 timesteps. a) has a low $\beta$ which means information cost is low while b) has a high $\beta$ and hence high cost on information}
	\label{fig:expres}
\end{figure}

%
%
%
%

Figure \ref{fig:expres} shows the probability distributions of the agent at a specific time $t=16$. Clearly, in the case where $\beta=0.5$, the agent is able to go through the region where the moving obstacle operates. However, when we increase the cost of information, the agent moves around the static obstacles.

\subsection{Static obstacles}\label{sec:marsexp}
Now, we present the example of the Mars rover navigating in the presence of static obstacles. Figure \ref{fig:mars} shows an example of a simple map of the environment that can obtained from a satellite image. This gives us a rough knowledge of the environment. We know the red region is impassable terrain, e.g. a jagged boulder. We also know that there is a region with a high density of obstacles and one with a low density of obstacles. All other regions are assumed to be obstacle free.

The helicopter can send information on the exact locations of obstacles to the rover to assist in path planning, however, we assign a cost to this information.

\begin{figure}[!t]

	\centering
\subfloat[Simple martian environment for case study. The red region is already known as impassable terrain. Additionally, two areas are identified as having a high and low probability of obstacles respectively\quad \quad \quad \quad \quad \quad \quad \label{fig:mars}]{
	\includegraphics[scale=0.12]{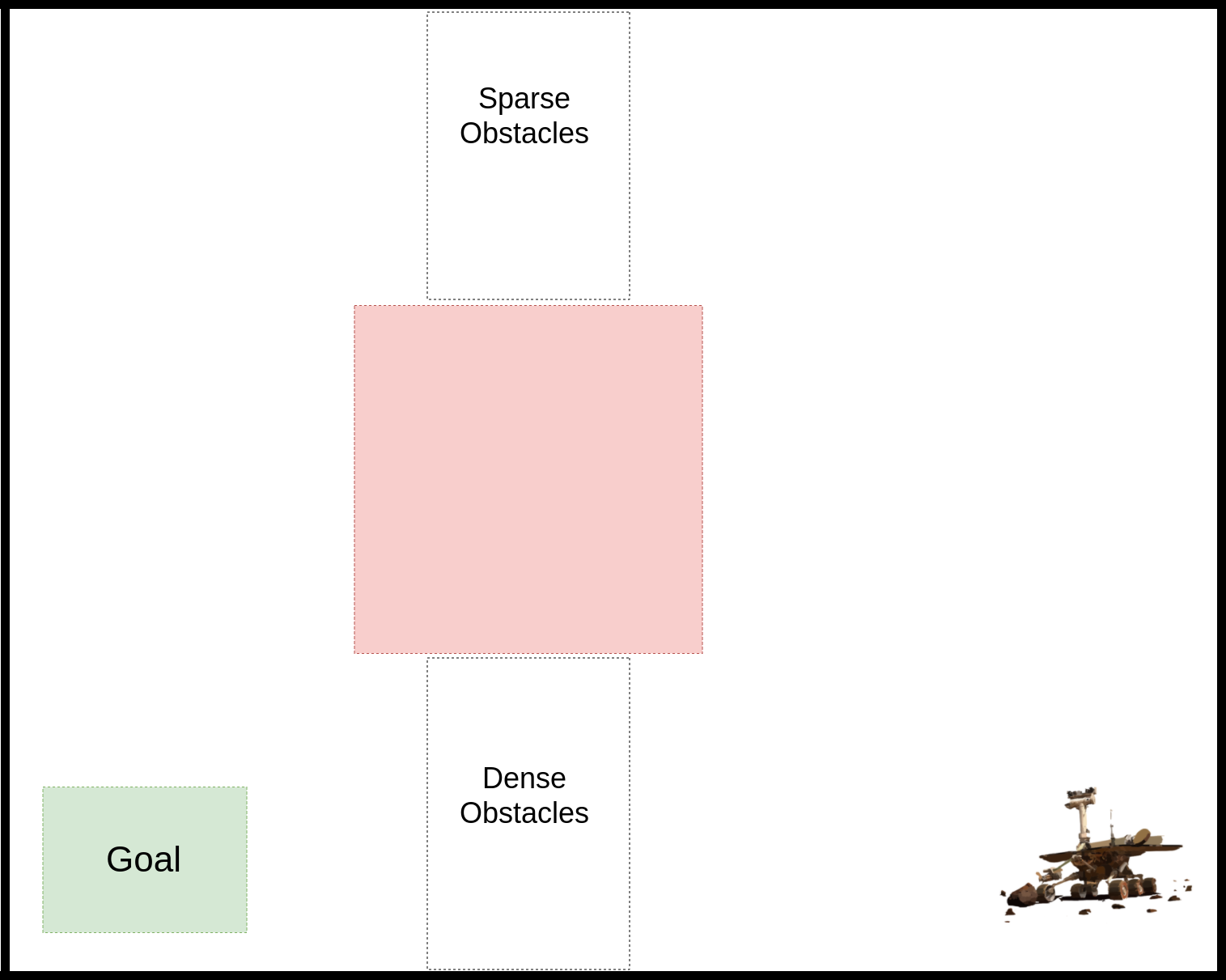}
}\\
\subfloat[True obstacle distribution of the scenario from \ref{fig:mars}. Red cells are obstacles, and the green cells represent the target region the rover shown in blue is trying to reach. \label{fig:mars_grid}]{
	\includegraphics[scale=0.4]{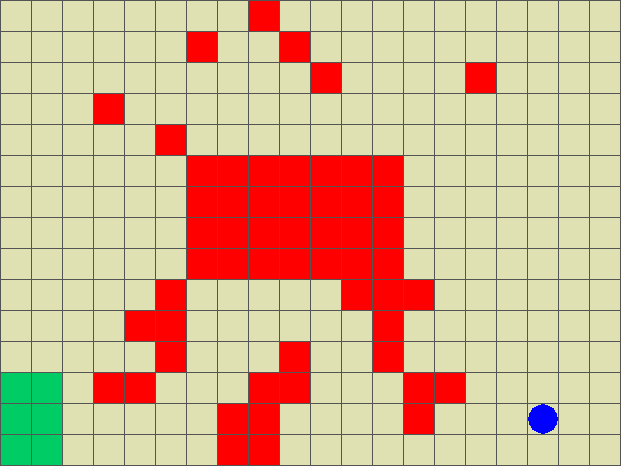}
}
\caption{We represent the environment in \ref{fig:mars} as a gridworld in \ref{fig:mars_grid} which we model as an MDP. 
	}
\end{figure}

The LTL specification is again $\lnot$'crash' $\LTLuntil$ 'goal' where the goal is the green region. We use a time horizon $T = 40$. This time the state in the MDP is given by $(x,y,(o_1,\dots,o_n))$ where $o_i \in [0,1]$ are probability values indicating the likelihood there is an obstacle in state $(x_i,y_i)$. We assign discrete values to $o_i$ by constraining it to values in the set $o_i \in [0,0.2,0.4,0.6,0.8,1]$. 

We model the helicopter flying ahead and scouting by allowing $o_i$ to transition to $0$ or $1$ with probability given by the value of $o_1$, if the rover is within distance $d$ of the obstacle. More explicitly, the state $(x_i,y_i,(o_1,\dots,o_j,\dots,o_n))$ will transition to $(x_i,y_i,(o_1,\dots,1,\dots,o_n))$ with probability $o_j$ and transition to $(x_i,y_i,(o_1,\dots,0,\dots,o_n))$ with probability $1 - o_j$. This will only happen if distance between the states $(x_i,y_i)$ and $(x_j,y_j)$ is less than or equal to a given range $d =2$.  $o_i = 1$ indicates there is an obstacle present in $(x_i,y_i)$. 

The region with sparse obstacle distribution has mostly $o_i = 0$ while the dense obstacle region has many more states with $o_i > 0$. This means that the rover will need the helicopter to scout ahead more often in the route with more obstacles. We assign the transfer entropy cost to states $(o_1,\dots o_n)$ which will penalize using these states in the policy synthesis.

\begin{figure}[!t]
\centering
	\includegraphics[scale=0.25]{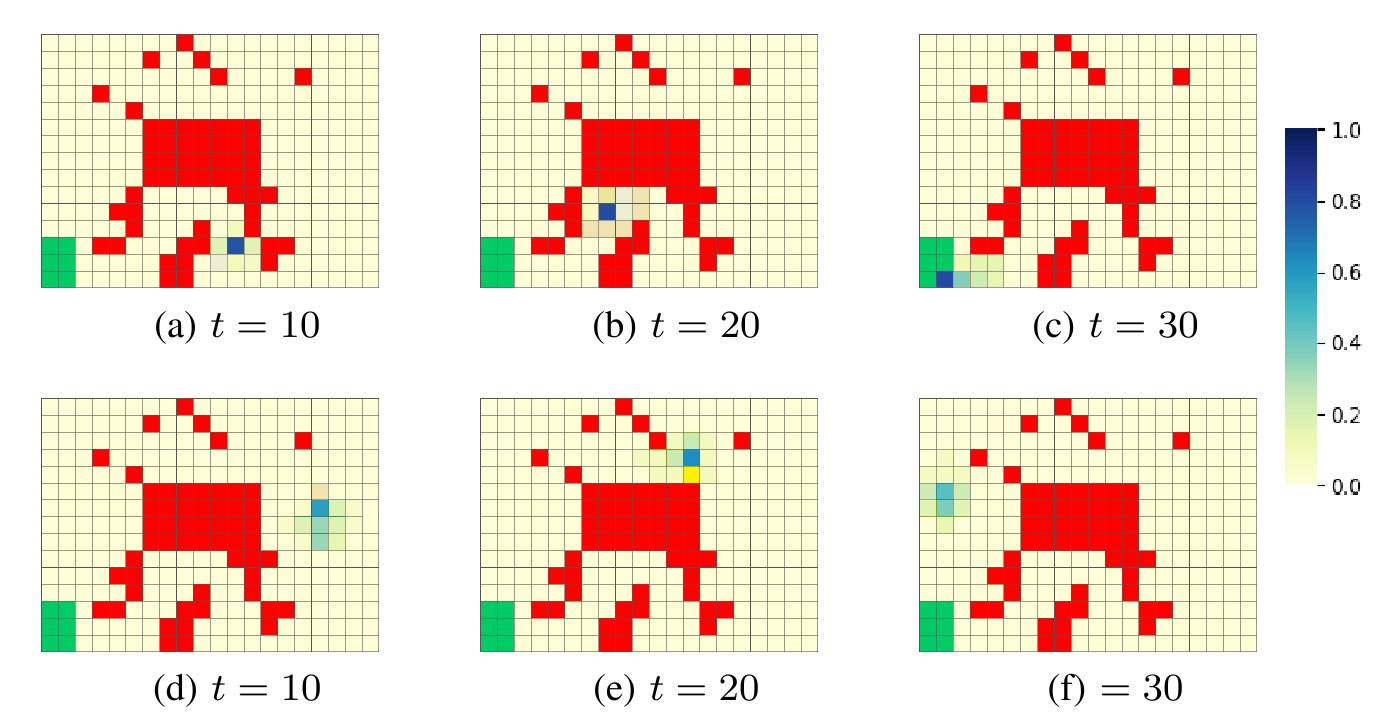}
	
	\caption{Probability distribution of the rover location over time for two cases: (a) - (c) : $\beta = 0$, (d) - (f): $\beta = 10$.
	}
	\label{fig:marsresults}
\end{figure}

Figure \ref{fig:marsresults} shows the evolution of the probability distribution of the agent when the information is free (i.e $\beta$ is small) and when information is expensive ($\beta$ is large). We see that when information is free, the rover takes the path through the dense obstacle distribution. Also note that since there is no information cost, the problem reduces to solving pure reachability and the policy is deterministic. When we set $\beta = 10$, the rover takes the path through the sparse obstacle region. Since there are fewer cells with non-zero probability of rocks, there is less need to sense for rocks and send the helicopter to scout $d$ states ahead The transfer entropy cost from $(o_1,\dots o_n)$ to $u$ is thus lower along the sparse obstacle path.


\section{Conclusion and Future Work}
In this paper, we presented a formal way to integrate co-safe LTL constraints into a minimal-information MDP problem. This is the first step in analyzing temporal logic constraints in communication constrained problems. For future work, we aim to relax the co-safe requirement to allow more general classes of LTL formulas by analyzing the mean information cost over an infinite run. Furthermore, we aim to extend this work to a multiple coordinating agent formulation as this problem setting naturally lends itself to minimizing communication between agents who are trying to satisfy a joint specification.

\bibliographystyle{IEEEtran}
\bibliography{main}

\appendices
\vspace{-0.15cm}
\section{Proof of Lemma 5.2} \label{sec:prf}
We will use the following basic result repeatedly.
\begin{lemma}
\label{lemblahut2} \cite[Theorem 4(c)]{blahut1972computation}
For fixed $\mu(x)$ and $\nu(u)$, the optimal solution to 
\[
\min_{q(u|x)} \sum_{\mathcal{X}}\sum_{\mathcal{U}} \mu(x)q(u|x)\left(\log\frac{q(u|x)}{\nu(u)}+c(x,u)\right)
\]
satisfies 
\[
q(u|x)=\frac{\nu(u)\exp\{-c(x,u)\}}{\sum_\mathcal{U}\nu(u)\exp\{-c(x,u)\}}
\]
$\mu(x)$-almost everywhere.
\end{lemma}

To prove Lemma~\ref{lem52}, it is sufficient to show the following statements hold for each $0\leq t\leq T-1$.
\begin{itemize}
\item[(a)] For fixed $q_0, ... , q_{t-1}$, the optimal solution to
\[
\min_{q_t} f(q_0, ... , q_{t-1}, q_t, q_{t+1}^*, ... ,q_{T-1}^*; \nu^*)
\]
satisfies \eqref{eqoptqt}.
\item[(b)] For fixed $q_{t+1}^*, ... , q_{T-1}^*$ satisfying \eqref{eqoptqt}, we have
\[
\sum_{k=t}^{T-1}\ell_k=-\sum_{\mathcal{X}^t}\sum_{\mathcal{U}^{t-1}}\mu_t(x^t, u^{t-1})\log \phi_t^*(x^t, u^{t-1}).
\]
\end{itemize}
We prove these statements by backward induction. For the time step $T-1$, notice that 
\begin{align*}
&f(q;\nu^*)= \text{(constant) } +\\
&\sum_{\mathcal{X}^{T-1}}\sum_{\mathcal{U}^{T-1}}\mu_{T-1}(x^{T-1}, u^{T-2})q_{T-1}(u_{T-1}|x^{T-1}, u^{T-2})\\
&\times \left(\log\frac{q_{T-1}(u_{T-1}|x^{T-1}, u^{T-2}}{\nu_{T-1}^*(u_{T-1}|\tilde{x}^{T-1}, u^{T-2})}+\rho_{T-1}^*(x^{T-1},u^{T-1})\right)
\end{align*}
where ``constant'' is the term that does not depend on $q_{T-1}$. Lemma~\ref{lemblahut2} is applicable to show that the minimizer $q_{T-1}^*$ satisfies \eqref{eqoptqt}. Statement (b) can be shown directly by substituting $q_{T-1}=q_{T-1}^*$ as
\begin{align}
&\ell_{T-1}\!=\!\!\sum_{\mathcal{X}^{T-1}}\!\sum_{\mathcal{U}^{T-1}}\!\mu_{T-1}(x^{T-1}\!\!, u^{T-2})q_{T-1}^*(u_{T-1}|x^{T-1}\!\!, u^{T-2}) \nonumber \\
&\hspace{2ex}\times \left(\log\frac{q_{T-1}^*(u_{T-1}|x^{T-1}, u^{T-2}}{\nu_{T-1}^*(u_{T-1}|\tilde{x}^{T-1}, u^{T-2})}+\rho_{T-1}^*(x^{T-1},u^{T-1})\right) \nonumber \\
&=\sum_{\mathcal{X}^{T-1}}\sum_{\mathcal{U}^{T-1}}\!\mu_{T-1}(x^{T-1}\!\!, u^{T-2})q_{T-1}^*(u_{T-1}|x^{T-1}\!\!, u^{T-2}) \nonumber \\
&\hspace{2ex}\times \left(-\log \phi_{T-1}^*(x^{T-1}, u^{T-2})\right) \nonumber\\
&=-\sum_{\mathcal{X}^{T-1}}\sum_{\mathcal{U}^{T-1}} \mu_{T-1}(x^{T-1}\!\!, u^{T-2})\log\phi_{T-1}^*(x^{T-1}, u^{T-2}) \nonumber \\
&\hspace{2ex}\times \underbrace{\sum\nolimits_{\mathcal{U}_{T-1}}q_{T-1}^*(u_{T-1}|x^{T-1},u^{T-1})}_{=1}. \label{eqqtsubstitute}
\end{align}

To complete the proof, we show that if (a) and (b) hold for the time step $t+1$, then they also hold for the time step $t$. Since (b) is hypothesized for $t+1$, using $\rho_t^*$, it is possible to write
\begin{align*}
&f(q_0, ... , q_t, q_{t+1}^*, ... , q_{T-1}^*;\nu^*) \\
&= \text{(constant) } +\sum_{\mathcal{X}^t}\sum_{\mathcal{U}^t}\mu_t(x^t, u^{t-1})q_t(u_t|x^t, u^{t-1})\\
&\hspace{13ex}\times \left(\log\frac{q_t(u_t|x^t, u^{t-1}}{\nu_t^*(u_t|\tilde{x}^t, u^{t-1})}+\rho_t^*(x^t,u^t)\right)
\end{align*}
where ``constant'' is the term that does not depend on $q_t$. Lemma~\ref{lemblahut2} is applicable once again to show that the minimizer $q_t^*$ satisfies \eqref{eqoptqt}.
Statement (b) for the time step $t$ can be shown by the direct substitution. Details are similar to \eqref{eqqtsubstitute}.

\end{document}